\documentclass[aps,pra,10pt,superscriptaddress,twocolumn,nofootinbib]{revtex4-2}
\usepackage{amsmath,amsthm,amssymb,dsfont,enumitem}
\usepackage{orcidlink}
\newcommand{\ket}[1]{\left| #1 \right\rangle}
\newcommand{\bra}[1]{\left\langle #1 \right|}
\newcommand{\av}[1]{\langle #1\rangle}
\usepackage{hyperref}
\usepackage[capitalise]{cleveref}
\crefname{equation}{Eq.}{Eqs.}
\crefname{appendix}{App.}{Apps.}
\crefname{figure}{Fig.}{Figs.}

\newcommand{\R}{\mathbb{R}}
\newcommand{\C}{\mathbb{C}}
\newcommand{\Hq}{\mathbb{H}}
\newcommand{\Oc}{\mathbb{O}}
\newcommand{\Cl}{\mathrm{Cl}}
\newcommand{\sog}{\mathfrak{so}}
\newcommand{\tr}{\mathrm{tr}}
\newcommand{\vx}{\mathbf{x}}
\newcommand{\vy}{\mathbf{y}}
\newcommand{\vn}{\mathbf{n}}
\newcommand{\vm}{\mathbf{m}}
\newcommand{\ve}{\mathbf{e}}
\newcommand{\vf}{\mathbf{f}}
\newcommand{\vs}{\mathbf{s}}
\newcommand{\vb}{\mathbf{b}}
\newcommand{\1}{\mathds{1}}
\newcommand{\Spin}{\mathrm{Spin}}

\theoremstyle{plain}
\newtheorem{theorem}{Theorem}
\newtheorem{proposition}[theorem]{Proposition}
\newtheorem{lemma}[theorem]{Lemma}
\theoremstyle{definition}
\newtheorem{definition}[theorem]{Definition}
\theoremstyle{remark}
\newtheorem{remark}[theorem]{Remark}

\begin{document}

\title{Why three? A two-level system with four mutually unbiased questions}
\author{Jonte R. Hance\,\orcidlink{0000-0001-8587-7618}}
\email{jonte.hance@newcastle.ac.uk}
\affiliation{Quantum Group, School of Computing, Newcastle University, 1 Science Square, Newcastle upon Tyne, NE4 5TG, UK}

\begin{abstract}
Two sharp yes/no questions are mutually unbiased if and only if their involutions anticommute. A set of unbiased questions is therefore a Clifford algebra, whose maximal size is odd, meaning four anticommuting questions typically imply a fifth: the two-level systems of real, complex and quaternionic quantum theory allow two, three and five questions, while four is skipped. Without this operator product restriction, a two-level system with four unbiased questions exists as the four-dimensional Bloch ball, the hyperbit. The single-system postulates that pick out the Jordan state spaces allow such a model, and only energy observability rules it out. We show what it lacks, and realise it inside two qubits: its Kirkwood-Dirac imaginary parts are hidden $\mathfrak{so}(4)$ generators, its composites allow entanglement but no interaction, and its 24-cell of states stays Spekkens preparation contextual down to depolarising strength $2/3$. In fermionic terms, the four questions are the Majorana operators of two modes while the fifth is parity: the four-question system is therefore the sector removed by parity superselection.
\end{abstract}

\maketitle

\textit{Introduction---}
A qubit has three mutually unbiased bases (MUBs), and a $d$-level system has at most $d+1$ \cite{Ivanovic1981,WoottersFields1989,Durt2010}. The usual explanation is the dimension count $(d^2-1)/(d-1)=d+1$ over orthogonal subspaces of the traceless Hermitian operators~\cite{Durt2010}, which for $d=2$ says the Bloch ball is three-dimensional. This answers ``how many'' given standard QM, but not ``why three dimensions'' more abstractly, a question the reconstruction programme has taken up~\cite{Hardy2001,DakicBrukner2011,MasanesMuller2011,MMPA2014,MullerMasanes2013,KrummMuller2019,BMU2014,Muller2021}.

We instead pose the question operationally. A two-level system supports yes/no questions, and two such questions are unbiased if a system prepared such that it answers one with certainty answers the other at random. A qubit supports three mutually unbiased questions, $\sigma_x,\sigma_y,\sigma_z$, constrained by $\av{\sigma_x}^2+\av{\sigma_y}^2+\av{\sigma_z}^2\le1$. Can a two-level system support four, and how much of a qubit would it still be?

In this paper we show that the answer has an algebraic, a geometric and a physical form. Algebraically, unbiased questions are anticommuting involutions, so a set of them is a Clifford algebra, and the product of four is a fifth: the operator-algebraic two-level systems allow $2$, $3$ or $5$ questions but never $4$. Geometrically, once this product relationship is dropped, a four-question system exists as the four-dimensional Bloch ball~\cite{PawlowskiWinter2012,JNW1934}, which we call the hyperbit.\footnote{Paw\l{}owski and Winter~\cite{PawlowskiWinter2012} use ``hyperbit'' for a ball of unbounded dimension, the qubit and classical bit included; we borrow it for the $n=4$ case.} This satisfies the single-system postulates short of energy observability, but lacks four properties possessed by a qubit: dynamics generated by observables, closure under products, a higher-level relative, and an interacting composite. Physically, the four questions are the Majorana operators of two fermionic modes, the fifth being fermion parity, and the hyperbit is the sector parity superselection forbids.

We then show the consequences of this. The hyperbit's 24-cell of states splits into three frames of mutually unbiased questions related by triality, and stays Spekkens preparation contextual down to depolarising strength $2/3$. Inside two qubits, the four questions grade the Pauli basis by $\Lambda\R^4$: the degree two elements being the generators, which appear as its Kirkwood-Dirac imaginary parts and are free to take either sign, and the failure of degree one to close under products is why its composites inherit entangled states without admitting interaction.

Throughout this paper, $n$ is a ball dimension and $d$ a Hilbert-space one: the qubit is $d=2,n=3$, the hyperbit $d=2,n=4$.

\textit{Unbiased questions and Clifford algebras---}
Let $A,B$ be Hermitian involutions on $\C^d$ ($A^2=B^2=\1$), each a sharp yes/no question with outcomes $\pm1$ and projectors $P^A_\pm=(\1\pm A)/2$.

\begin{definition}
$A$ and $B$ are \emph{mutually unbiased} if every state in an eigenspace\footnote{Note this condition is on whole eigenspaces, not just one eigenbasis: for degenerate $A$ the weaker reading is not enough.On $\C^4$, $A=Z\otimes\1$ and $B=(Z\otimes X+X\otimes\1)/\sqrt2$ are traceless Hermitian involutions, every computational basis vector gives $B$ each outcome with probability $1/2$, and yet $AB+BA\neq0$.} of $A$ gives each outcome of $B$ with probability $1/2$, and vice versa.
\end{definition}

This is the two-outcome case of a mutually unbiased measurement~\cite{TavakoliMUM2021,FarkasKaniewskiNayak2023}, for which the second condition is also usually imposed, though \cref{lem:anticomm} shows the first implies it.

\begin{lemma}
\label{lem:anticomm}
$A$ and $B$ are mutually unbiased iff $AB+BA=0$.
\end{lemma}

\begin{proof}
If $\bra\psi B\ket\psi=0$ on the $+1$ eigenspace of $A$, the Hermitian operator $P^A_+BP^A_+$ has vanishing quadratic form there and is zero; likewise $P^A_-BP^A_-=0$. Then $B=P^A_+BP^A_-+P^A_-BP^A_+$ and $AB=-BA$. Conversely, if $AB=-BA$ and $A\psi=\psi$ then $\bra\psi B\ket\psi=\bra\psi BA\ket\psi=-\bra\psi AB\ket\psi=-\bra\psi B\ket\psi$.
\end{proof}

Wehner and Winter showed anticommutation makes binary observables unbiased \cite{WehnerWinter2008,WehnerWinter2010}; the converse makes the two properties one. A set of pairwise unbiased questions is thus a set of pairwise anticommuting involutions $e_1,\dots,e_n$, a spin system in the Jordan-algebraic sense~\cite{HancheOlsenStormer1984,BGW2020} and a representation of the Clifford algebra $\Cl_n$ with $e_ie_j+e_je_i=2\delta_{ij}\1$, whose operator basis is graded by $\Lambda\R^n$~\cite{WehnerWinter2008}.

\begin{proposition}
\label{prop:count}
Write $d=m2^q$ with $m$ odd. The maximal number of pairwise unbiased sharp yes/no questions on $\C^d$ is $2q+1$, always odd: $3$ for $d=2$, $5$ for $d=4$, $7$ for $d=8$, $1$ for odd $d$.
\end{proposition}

\begin{proposition}
\label{prop:completion}
For pairwise anticommuting involutions $e_1,\dots,e_n$ the product $p_n=e_1\cdots e_n$ is a further Hermitian involution anticommuting with all $e_i$, and so a further unbiased question, if and only if $n\equiv0\pmod4$.
\end{proposition}

\begin{proof}
Reversion gives $p_n^\dagger=(-1)^{n(n-1)/2}p_n$ and $p_n^2=(-1)^{n(n-1)/2}\1$, so $p_n$ is a Hermitian involution for $n\equiv0,1\pmod4$; and moving $e_i$ through $p_n$ passes $n-1$ anticommuting factors, so $p_n$ commutes with every $e_i$ for odd $n$ and anticommutes for even $n$.
\end{proof}

Wehner and Winter \cite{WehnerWinter2008,WehnerWinter2010} already count the pseudoscalar as a further anticommuting observable, giving $2k+1$ of them on $k$ qubits; the sign rule fixes which product it is, since their two papers take $p_n$ with and without a factor of $i$, and each is Hermitian only where the other is not. The real Clifford algebras for $n=2,\dots,5$ are $\Cl_n=M_2(\R),M_2(\C),M_2(\Hq),M_2(\Hq)^{\oplus2}$: at $n=2,3$ the bare product squares to $-\1$, generating rotations for the rebit and supplying the imaginary unit for the qubit; at $n=5$ it is central and merely labels the two irreducible representations. Only at $n\equiv0\pmod4$ is $p_n$ itself a further question; at $n\equiv2\pmod4$ it is $ip_n$ that is. Either way a set of four never closes, since by \cref{prop:count} a maximal set is odd: ``four unbiased questions'' completes to the quaternionic bit's five, and the two-level systems of real, complex and quaternionic quantum theory are the $2$-, $3$- and $5$-balls.

For $d>2$ questions and MUBs differ, a basis having $d$ outcomes and a question two; for $d=2$ they coincide, and everything below concerns questions.

\textit{The $n$-ball---}
Let us drop the product relationship and keep the convex structure. States are $\omega=(1,\vx)$ with $\vx\in\R^n$, $|\vx|\le1$; effects are $E=(a,\vb)$ with $E(\omega)=a+\vb\cdot\vx$ and $|\vb|\le a\le1-|\vb|$, the nontrivial extremal ones being $E^\pm_\vn=\tfrac12(1\pm\vn\cdot\vx)$ for unit $\vn$, which pair into the sharp question along $\vn$. Reversible transformations are $O(n)$, with identity component $SO(n)$.

The single-system features usually attributed to the qubit follow from this. Questions along $\vn$ and $\vm$ are unbiased iff $\vn\cdot\vm=0$, so a maximal unbiased set is an orthonormal frame with $n$ members, and the pure state $\vn$ answers the question along $\vm$ with $p(E^+_\vm)=\cos^2(\theta/2)$. Writing $\av{e_i}=x_i$, the Brukner-Zeilinger invariant $\sum_i\av{e_i}^2\le1$ holds with equality on pure states~\cite{BruknerZeilinger1999,BruknerZeilinger2009,PDB2010}: one bit shared among $n$ complementary questions. Perfectly distinguishable states are antipodal, so the information capacity is one bit, regardless the value of $n$~\cite{Massar2015}.

Other $B_n$-symmetric bodies carry the same $n$ questions, the cross-polytope $\sum_i|x_i|\le1$ among them; the ball is the only one whose pure states can be moved by a continuous group (\cref{prop:whyball}). This is the role of Ref.~\cite{MasanesMuller2011}'s symmetry requirement, called continuous reversibility in~\cite{Muller2021}.

The $n$-ball is the state space of the spin factor $V_n=\R\oplus\R^n$ with Jordan product $(s,\vx)\circ(t,\vy)=(st+\vx\cdot\vy,\,s\vy+t\vx)$ \cite{JNW1934}, a rank-two Euclidean Jordan algebra, with $V_2\cong H_2(\R)$, $V_3\cong H_2(\C)$, $V_5\cong H_2(\Hq)$ and $V_9\cong H_2(\Oc)$. Barnum, M\"uller and Ududec~\cite{BMU2014} showed that classical decomposability, strong symmetry and absence of third-order interference leave exactly the irreducible Jordan systems and the classical ones, a list including balls of every dimension: at the single-system level these postulates do not separate $n=4$ from $n=3$.

\textit{What the hyperbit lacks---}
The qubit possesses four properties the hyperbit lacks (\cref{tab:losses}). The first comes from dimension count; the other three come from closure.

\begin{enumerate}[label=(\Alph*)]
\item \emph{Observables do not generate dynamics:} $\dim\sog(n)=n(n-1)/2$ equals $n$ only for $n=3$: three is the only dimension in which the rotations of the ball can be labelled by points on the ball. The qubit's Hodge star identifies $\Lambda^2\R^3$ with $\R^3$, so every observable generates a rotation and every rotation is generated by an observable; the hyperbit's generators are the six bivectors $\ve_i\wedge\ve_j$, none of which are observables, and $\sog(4)\cong\mathfrak{su}(2)\oplus\mathfrak{su}(2)$ acts on the pure-state sphere $S^3$ by quaternion multiplication on both sides. Observability of energy, Barnum, M\"uller and Ududec's fourth postulate~\cite{BMU2014} which cuts the Jordan-algebraic list down to complex quantum theory, asks for an injection of generators into observables and so compares $n(n-1)/2$ with $\dim V_n=n+1$, which already excludes every $n\ge4$; they show it holds on the $n$-ball only for $n=3$. Alfsen and Shultz's dynamical correspondence~\cite{AlfsenShultz1998,AlfsenShultz2003} instead exists exactly when the algebra is the self-adjoint part of a C*-algebra, so not for the simple five-dimensional $V_4$.\label{PropertyAObsNoDynam}

\item \emph{No algebraic closure:} A Jordan algebra embedded in an associative algebra is \emph{reversible} if it is closed under the symmetrised products $a_1\cdots a_k+a_k\cdots a_1$, and $V_n$ has a reversible representation only for $n\in\{2,3,5\}$~\cite{HancheOlsen1983,HancheOlsenStormer1984,BGW2020}. \cref{prop:completion} obstructs this for $n=4$: $e_1e_2e_3e_4+e_4e_3e_2e_1=2p_4$ is a fifth question, and the smallest reversible Jordan algebra containing $V_4$ is $V_5=H_2(\Hq)$.\label{PropertyBNoAlgClos}

\item \emph{No higher-level relatives:} The rank-two faces of the simple Euclidean Jordan algebras of rank at least three are $V_2,V_3,V_5,V_9$ \cite{JNW1934,BGW2020}, and spin factors have no higher rank, so no three-level system satisfying Ref.~\cite{BMU2014};s postulates contains the hyperbit as a two-level subspace. This is a failure of Hardy's subspace axiom \cite{Hardy2001}.\label{PropertyCNoHLRels}

\item \emph{No interacting composites:} Masanes \emph{et al.}~\cite{MMPA2014} found no entangled states, or equivalently no interaction, for two $n$-balls with $n\neq3$ under tomographic locality, strict convexity and continuous reversibility; Krumm and M\"uller~\cite{KrummMuller2019} sharpened this to circuits, where for $n\neq3$ every gate is a product of single-wire gates.\label{PropertyDNoIntComps}
\end{enumerate}

Property~\ref{PropertyAObsNoDynam} fails for every $n\neq3$; local tomography fails for real and quaternionic theory too, despite both having interacting composites \ref{PropertyDNoIntComps}, though the composite of two quaternionic systems is a real one~\cite{Baez2012,BGW2020}. $n=2,3,5$ have \ref{PropertyBNoAlgClos} and \ref{PropertyCNoHLRels}, but $n=4$ has neither: a rank-$k$ parent supplies the associative product whose symmetrised form closes the algebra.

\begin{table}[t]
\caption{The four qubit properties the hyperbit lacks: \ref{PropertyAObsNoDynam} observables generate dynamics, \ref{PropertyBNoAlgClos} reversible closure, \ref{PropertyCNoHLRels} the division algebra $\mathbb F$ of a higher-rank parent $H_k(\mathbb F)$, \ref{PropertyDNoIntComps} an interacting composite, \cite{BGW2020}'s canonical property. LT marks local tomography of that composite, except in the hyperbit's row, which refers to \cref{prop:compositeprops}'s composite. $V_9$ has a rank-three parent only; where \ref{PropertyBNoAlgClos} fails the canonical composite enlarges the one-sided observables, to $V_5$ for the hyperbit.}
\label{tab:losses}
\begin{ruledtabular}
\begin{tabular}{lccccc}
$n$, system & (A) & (B) & (C) & LT & (D) \\
\hline
2, rebit & no & yes & $\R$ & no & yes \\
3, qubit & yes & yes & $\C$ & yes & yes \\
4, hyperbit & no & no & none & yes & enlarged \\
5, $\Hq$-bit & no & yes & $\Hq$ & no & yes \\
9, $\Oc$-bit & no & no & $\Oc$ & no & enlarged \\
\end{tabular}
\end{ruledtabular}
\end{table}

\textit{Polytopes, quasiprobabilities and contextuality---}
Two features of the hyperbit have no qubit analogue: the geometry of its frames, and the Spekkens contextuality of the subtheory they generate.

For the qubit the discrete structure is the octahedron of MUB eigenstates and the cube of magic states; the hyperbit's analogues are larger~\cite{Coxeter1973}. Write $\ve_1,\dots,\ve_4$ for the standard basis and $\vs\in\{\pm1\}^4$ for sign strings. Frame $\mathcal A=\{\pm\ve_i\}$ is the 16-cell of eigenstates of a maximal unbiased set, and the Hadamard states $\vs/2$ are the tesseract, at overlap $\pm\tfrac12$ with every axis. Its even- and odd-parity halves $\mathcal B=\{\vs/2:\prod_is_i=+1\}$ and $\mathcal C=\{\vs/2:\prod_is_i=-1\}$ are each a 16-cell too. The union $\mathcal A\cup\mathcal B\cup\mathcal C$ is the 24-cell, the binary tetrahedral group of unit quaternions.

\begin{proposition}
\label{prop:24cell}
$\mathcal A,\mathcal B,\mathcal C$ are orthonormal frames of $\R^4$ whose union is the 24-cell; vectors from different frames have dot product $\pm\tfrac12$. The symmetry group $W(F_4)$ of the 24-cell, of order $1152$, permutes the three frames as the full $S_3$; the hyperoctahedral group $B_4$, of order $384$, is the stabiliser of a single frame, and the subgroup fixing all three is $W(D_4)$, of order $192$, so the permutation action is $D_4$ triality. The 24-cell contains sixteen regular hexagons, each taking one antipodal pair from each frame.
\end{proposition}

The frame statements are a count of sign agreements: $\ve_i\cdot\vs/2=s_i/2$, distinct strings of equal parity differ in an even number of places and give $\vs\cdot\vs'/4\in\{0,-1\}$, and strings of opposite parity differ in an odd number and give $\pm1/2$. The group statements are standard \cite{Coxeter1973}, and the hexagon count and its frame composition were verified by enumeration. One hexagon, in the plane of $\ve_1$ and $(0,1,1,1)/\sqrt3$, is
\begin{equation}
\label{eq:hexagon}
\begin{split}
\ve_1,\ \tfrac12(+{+}{+}+),\ \tfrac12(-{+}{+}+),\\ -\ve_1,\ \tfrac12(-{-}{-}-),\ \tfrac12(+{-}{-}-),
\end{split}
\end{equation}
alternating between $\mathcal A$, $\mathcal B$ and $\mathcal C$.

Each frame also gives a self-dual quasiprobability representation, $W(\vx)=(1+2\,\vn\cdot\vx)/8$ at each of its eight vectors $\vn$, the $n=4$ counterpart of Wootters' discrete Wigner function \cite{Wootters1987,GHW2004} (see Appendix~\ref{app:wigner}). One frame's representation is nonnegative exactly on the other two frames and their hull (\cref{thm:wigner}), so triality exchanges the roles of stabiliser and magic states; for the qubit that set is the cube, the magic states of the stabiliser polytope.

The cross-polytope subtheory of the $n$-ball, its states $\pm\ve_i$ with the $n$ axis questions and their mixtures, has a preparation- and measurement-noncontextual model on $\{\pm1\}^n$ (\cref{prop:16cell}); for $n=3$ it is an eight-state relative of Spekkens' toy bit~\cite{Spekkens2007}. The 24-cell behaves differently.

\begin{theorem}
\label{thm:24contextual}
The 24-cell subtheory of the hyperbit, with the $24$ states of $\mathcal A\cup\mathcal B\cup\mathcal C$, the $12$ axis questions of the three frames and their mixtures, is preparation contextual.
\end{theorem}

\begin{proof}
A regular hexagon of pure states gives Spekkens' configuration~\cite{Spekkens2005}: alternate vertices $a,b,c$ and their antipodes $A,B,C$, each antipodal pair perfectly distinguishable and mixing to the maximally mixed state, to which the two alternating triples $\{a,b,c\}$ and $\{A,B,C\}$ also mix. He shows that no preparation-noncontextual model reproduces it, and the 24-cell contains sixteen such hexagons (\cref{prop:24cell}), one of which is \cref{eq:hexagon}.
\end{proof}

The hexagon also obeys a tight inequality: with $s$ the average probability of telling the antipodal pair $P_k,P_K$ apart with its own question $M_k$,
\begin{equation}
\label{eq:sbound}
s=\tfrac16\textstyle\sum_{k\in\{a,b,c\}}\big[p(+|M_k,P_k)+p(-|M_k,P_K)\big]\le\tfrac56
\end{equation}
in every model preparation-noncontextual for all of the hexagon's equivalences, the three-element ones included (see Appendix~\ref{app:proofs}). Mazurek \emph{et al.}~\cite{Mazurek2016} assume only the pairwise equivalences, adding measurement noncontextuality, which their Appendix B.2 shows is then needed; they test the resulting bound experimentally. Depolarising the hexagon, $\vx\mapsto\lambda\vx$, the ball gives $s=\tfrac12(1+\lambda)$, so the hexagon subtheory is contextual for $\lambda>\tfrac23$, and the threshold is tight: \cite{Mazurek2016}'s three-state model reproduces the depolarised hexagon at $\lambda=2/3$, satisfies the further equivalences too, and is noncontextual for preparations and measurements alike.

The 24-cell subtheory is therefore preparation contextual at least for $\lambda>2/3$; whether its further equivalences lower the threshold, and whether it is measurement contextual, we leave open, and both are linear programs~\cite{Schmid2021,Selby2024}. The threshold is not the hyperbit's: $s$ depends only on the hexagon's plane, so $2/3$ is common to every $n$-ball, corresponding to M\"uller and Garner's $\varepsilon<\tfrac16$ for the rebit in $\varepsilon$-embedding language~\cite{MullerGarner2023}. What changes at $n=4$ is where the hexagon sits: the qubit's octahedron of MUB eigenstates is noncontextual (\cref{prop:16cell}) and contains none, while the hyperbit's 24-cell contains sixteen.

\textit{The hyperbit inside two qubits}
The hyperbit cannot be fundamental, but can be represented using other systems. The unrestricted models that embed in finite-dimensional quantum theory are the Euclidean \emph{special} Jordan algebras~\cite{MullerGarner2023}, everything on the Jordan list but $H_3(\Oc)$, and the spin factor $V_n$ embeds in $\lfloor n/2\rfloor$ qubits by Tsirelson's construction \cite{Tsirelson1987,Kleinmann2013,BGW2020}. For $n=4$ that is two qubits: take the four anticommuting Pauli operators
\begin{equation}
\label{eq:eops}
e_1=X\otimes X,\quad e_2=Y\otimes X,\quad e_3=Z\otimes X,\quad e_4=\1\otimes Y,
\end{equation}
with $e_5=\1\otimes Z=-e_1e_2e_3e_4$ \cref{prop:completion}'s fifth question, and restrict the observables to $\mathcal O_4=\mathrm{span}_\R\{\1,e_1,\dots,e_4\}$, writing $\vn\cdot\ve=\sum_in_ie_i$.

\begin{proposition}[The standard embedding]
\label{prop:projection}
$\pi(\rho)=(\tr\rho e_1,\dots,\tr\rho e_4)$ maps the two-qubit states onto the closed unit $4$-ball; the question along a unit $\vn$ is the rank-two projective measurement $\{\tfrac12(\1\pm\vn\cdot\ve)\}$, with $\tr(\rho\,\tfrac12(\1\pm\vn\cdot\ve))=E^\pm_\vn(\pi(\rho))$; and conjugation by $\exp(\theta e_ie_j)$ and by $e_i$ implements $O(4)$.
\end{proposition}

\begin{table}[t]
\caption{The two-qubit Pauli basis graded by the Clifford algebra of the four questions~\cite{WehnerWinter2008}, signs and factors of $i$ omitted from the Pauli labels, with the fermionic reading given below. The hyperbit sees grades $0$ and $1$; the even grades form the parity-preserving algebra, closed under products.}
\label{tab:grading}
\begin{ruledtabular}
\begin{tabular}{llll}
grade & operators & hyperbit & fermionic \\
\hline
$\Lambda^0$ & $II$ & unit & $\1$ \\
$\Lambda^1$ & $XX,YX,ZX,IY$ & questions $e_i$ & $\gamma_i$ \\
$\Lambda^2$ & $ZI,YI,XI,XZ,YZ,ZZ$ & generators $G_{ij}$ & $i\gamma_j\gamma_i$ \\
$\Lambda^3$ & $XY,YY,ZY,IX$ & $ie_5e_i$ & $i\gamma_i\gamma_j\gamma_k$ \\
$\Lambda^4$ & $IZ$ & fifth question $e_5$ & parity \\
\end{tabular}
\end{ruledtabular}
\end{table}

\Cref{tab:grading} grades the fifteen traceless two-qubit Pauli operators by $\Lambda^1\oplus\Lambda^2\oplus\Lambda^3\oplus\Lambda^4$ of $\R^4$, of dimensions $4+6+4+1$; degree two holds the six generators $G_{ij}=ie_je_i$. For the qubit, $\Lambda^2$ and $\Lambda^3$ of $\Cl_3$ are anti-Hermitian, so its observables are $\Lambda^0\oplus\Lambda^1$ and it sees everything: the Hodge star folds its degree two back into observables. There is no such identification at $n=4$.

The two-qubit realisation also underlies the operational claims made for hyperbits. Paw\l{}owski and Winter~\cite{PawlowskiWinter2012} showed that for tasks in which the receiver returns one bit, any strategy using one hyperbit can be simulated by shared entanglement and one classical bit, by Tsirelson's construction~\cite{Tsirelson1987,Tsirelson1980} with the hyperbit directions as Clifford generators. The other direction of their claimed equivalence fails, as they have since acknowledged \cite{PawlowskiWinterErratum}: entanglement-assisted classical bits are strictly stronger than hyperbits~\cite{Tavakoli2021}, the equivalence holding only if the receiver knows the hyperbit expectation value in advance~\cite{Scala2024}.

\textit{Kirkwood-Dirac distributions}
Kirkwood-Dirac (KD) distributions \cite{Kirkwood1933,Dirac1945,ArvidssonShukur2024} need products of projectors, which the hyperbit lacks but the two-qubit realisation possesses. For anticommuting involutions $A,B$ and any state $\rho$, since $P^B_bP^A_a=\tfrac14(\1+aA+bB+ab\,BA)$ and $BA=-iG$ with $G=iBA$,
\begin{equation}
\label{eq:kd}
Q_\rho(a,b)=\tr(P^B_bP^A_a\rho)=\tfrac14\big[1+a\av A+b\av B-i\,ab\,\av G\big],
\end{equation}
where $G$ is a Hermitian involution anticommuting with $A$ and $B$. Take a qubit with $A=\sigma_x$ and $B=\sigma_z$, so that $G=-\sigma_y$: the imaginary part of the KD distribution of two unbiased questions is the third, with a sign fixed by the orientation of the Bloch sphere, $G_{ij}=\epsilon_{ijk}e_k$. For the hyperbit $G_{ij}=ie_je_i$ are \cref{tab:grading}'s six degree-two operators, none of them a hyperbit observable.

\begin{theorem}
\label{thm:kdrange}
For a hyperbit state $\vx$ and $i\neq j$, among two-qubit states $\rho$ with $\pi(\rho)=\vx$ the KD imaginary part of the pair $(e_i,e_j)$ is $-\tfrac14ab\,g_{ij}$ with $g_{ij}=\tr(\rho G_{ij})$, and $g_{ij}$ ranges over the whole interval
\begin{equation}
\label{eq:kdrange}
|g_{ij}|\le\sqrt{1-x_i^2-x_j^2}.
\end{equation}
Unless $x_i^2+x_j^2=1$ it is therefore not a function of the hyperbit state.
\end{theorem}

\begin{remark}[Symmetry version]
\label{rem:equivariant}
At $\vx\neq0$ a covariant rule $g(R\vx)=R\,g(\vx)R^{\mathsf T}$ assigning a bivector to each state must take a value invariant under the stabiliser $SO(n-1)$ of $\vx$. For $n=4$, $\Lambda^2\R^4=(\vx\wedge\vx^\perp)\oplus\Lambda^2\vx^\perp\cong\R^3\oplus\R^3$ as $SO(3)$-representations, with no invariant vector, so $g\equiv0$. For $n=3$ the stabiliser is $SO(2)$ and $\Lambda^2\vx^\perp$ is an invariant line, the Hodge dual of $\vx$, which is the rule the qubit uses. The two-qubit substrate shows that every value in \cref{eq:kdrange} occurs.
\end{remark}

The real part of \cref{eq:kd}, $\tfrac14(1+ax_i+bx_j)$, is fixed by the hyperbit and can be negative, just as for a qubit. The imaginary part instead decouples from the state, because a hyperbit state carries no orientation: choosing the identification $\Lambda^2\R^3\cong\R^3$ is choosing the qubit's $i$, whereas $\Lambda^2\R^4$ is six-dimensional and admits none. A qubit's KD imaginary part is a statement about the third question; the hyperbit's is about a generator the theory cannot measure and the underlying two-qubit space is free to set. The imaginary part can be interpreted as a generator~\cite{DresselJordan2012}, here generating a rotation the hyperbit cannot see, so arguments using KD imaginarity or anomalous weak values as a witness of a property of the state (see~\cite{ArvidssonShukur2024} for a survey) just witness the underlying two-qubit space. Schmid \emph{et al.}~\cite{Schmid2024} make the companion point inside quantum theory, that KD negativity or imaginarity does not by itself imply contextuality.

\textit{Composites}
Ref.~\cite{MMPA2014,KrummMuller2019}'s no-go results assume the composite state space is reached from product states by its own reversible transformations, or that these act transitively on pure states. The two-qubit realisation suggests a composite that drops this, inheriting its entangled states from the ambient four qubits. Take $e_1,\dots,e_5$ on qubits $1,2$, and $f_1,\dots,f_5$ likewise on qubits $3,4$, writing $\vm\cdot\vf=\sum_jm_jf_j$, and restrict the observables to
\begin{equation}
\label{eq:compositespan}
\mathcal O_{4\otimes4}=\mathrm{span}_\R\{\1,\ e_i\otimes\1,\ \1\otimes f_j,\ e_i\otimes f_j\},
\end{equation}
with $1\le i,j\le4$, of dimension $25$. States are the images $(\vx,\vy,C)$ of four-qubit states, with $x_i=\av{e_i}$, $y_j=\av{f_j}$, $C_{ij}=\av{e_i\otimes f_j}$.

\begin{proposition}
\label{prop:compositeprops}
The composite \cref{eq:compositespan} is locally tomographic and no-signalling. Product states have $C=\vx\vy^{\mathsf T}$, and as for the Bloch correlation matrix of a bipartite quantum state \cite{deVicente2007} every separable state has $\|C\|_1\le1$, so $\|C\|_1>1$ witnesses entanglement. There is a unique state with $C=\1_4$ and $\vx=\vy=0$, the joint $+1$ eigenstate of the four commuting involutions $e_i\otimes f_i$; it has $\|C\|_1=4$, the maximum, and $\av{e_5\otimes f_5}=1$, and with $e_1,e_3$ against $(f_1\pm f_3)/\sqrt2$ its CHSH value is $2\sqrt2$, which no state exceeds.
\end{proposition}

A Bell pair of qubits has $C=\mathrm{diag}(1,-1,1)$ and $\|C\|_1=3$: the hyperbit pair is more correlated in trace norm without exceeding Tsirelson's bound, which sees only a $2\times2$ block of $C$. Whether $C=\1_4$ self-tests four anticommuting involutions we leave open.

\begin{theorem}
\label{thm:localonly}
The Lie algebra of the four-qubit unitaries $U$ with $U\mathcal O_{4\otimes4}U^\dagger=\mathcal O_{4\otimes4}$ is $i\R\1\oplus\sog(4)\oplus\sog(4)$, of dimension $13$, spanned by $i\1$ and the local bivectors $e_ie_j\otimes\1$, $\1\otimes f_kf_l$. Every continuous reversible transformation of the composite that keeps the effective description closed is local.
\end{theorem}

The proof reduces to a combinatorial test on the $256$ four-qubit Pauli operators (\cref{lem:pauli}).

This obstruction comes from closure: an interaction keeping the effective description closed needs a local observable space closed under the operator product, and $\mathrm{span}\{\1,e_1,\dots,e_4\}$ never is, since it would have to contain $e_1e_2$. The one space in this family with an entangling generator is the even subalgebra $\Lambda^0\oplus\Lambda^2\oplus\Lambda^4\cong M_2(\C)\oplus M_2(\C)$, closed under products and not a hyperbit at all but a qubit with a classical bit. The same test disposes of realisations with multiplicity and of enlargements of the local observable space, and a product generator is ruled out by hand (\cref{app:local}).

The composite \cref{eq:compositespan} thus has entanglement without interaction, as in boxworld, where the reversible dynamics is generated by local operations and permutations alone~\cite{Barrett2007,GMCD2010}.

Letting the composite's observables exceed the span of the local ones does not rescue the hyperbit either, and the reason is again closure. Barnum, Graydon and Wilce's canonical tensor product $A\odot B$, the Jordan algebra generated by the products $a\otimes b$, is a composite of any two embedded Euclidean Jordan systems, with no reversibility asked of the factors \cite{BGW2020}. Take two copies of $V_4$ in the standard embedding. The one-sided observables of $V_4\odot V_4$ are $V_5$ on each side, not $V_4$, and $V_4\odot V_4=V_5\odot V_5$ as subspaces of $M_{16}(\C)$: the composite is the quaternionic bit's, $M_{16}(\R)_{\rm sa}$ \cite{BGW2020}, of dimension $136$ against the $25$ of \cref{eq:compositespan}, so neither locally tomographic nor local in its dynamics. A partner system hands the fifth question back, and it does so where (B) fails: $V_5$ is the smallest reversible Jordan algebra containing $V_4$, and that is the algebra the partner sees.

\textit{The hyperbit as the odd sector of two fermionic modes---}
Four pairwise anticommuting involutions on $\C^4$ have a standard name: the Majorana operators $\gamma_1,\dots,\gamma_4$ of two fermionic modes, the factorisation in the proof of \cref{thm:kdrange} being their Jordan--Wigner representation~\cite{BravyiKitaev2002}. Every structure above then has a fermionic name, in the last column of~\cref{tab:grading}: the fifth question is parity up to sign, the state $\vx$ the one-point functions $\av{\gamma_i}$, the generators $G_{ij}$ the covariance $i\av{\gamma_j\gamma_i}$, and $SO(4)$ the Gaussian unitaries.

The degree-two operators are the Pauli operators of a Majorana qubit \cite{Nayak2008,BravyiTerhal2010}, with $X,Y,Z$ proportional to $G_{23},G_{31},G_{12}$; the two parity sectors carry the two $\mathfrak{su}(2)$ summands of $\sog(4)$, self-dual bivectors vanishing on one and anti-self-dual on the other. They span the even subalgebra given in our composite discussion, the physical observables under the superselection rule.

The hyperbit's observables are instead the degree-one operators, the single Majorana operators, which superselection forbids~\cite{Kitaev2001,Szalay2021}: the hyperbit is the sector superselection removes, and a Majorana qubit is its complement. For every state commuting with parity the one-point functions vanish, since $\gamma_i$ anticommutes with the parity $e_5$, so every physical state of two modes is the maximally mixed hyperbit state and the hyperbit's pure states are parity-violating superpositions. What a physical state does have, \cref{thm:kdrange}'s covariance matrix $g_{ij}=i\av{\gamma_j\gamma_i}$ together with the parity $\av{e_5}$, is what the hyperbit cannot see.

We now have a physical reason for three of the four missing properties given above. Observables do not generate dynamics: fermionic Hamiltonians are even, sums of bilinears, never linear in a Majorana operator. There is no reversible closure: odd operators do not close under products, their associative closure being the whole algebra and their reversible Jordan closure $V_5$. There is no interacting composite: couplings between fermionic systems are even overall, the hopping terms $i\gamma_i\gamma'_j$, and fermionic systems compose by the graded tensor product, which fails local tomography~\cite{DAriano2014a,DAriano2014b}. This places the hyperbit in the sector the graded composite excludes: the hopping generator $e_i\otimes f_j$ lies in the composite's span, \cref{eq:compositespan}, but conjugating $e_k\otimes\1$ by it gives $e_ke_i\otimes f_j$, three Majorana operators, odd overall and outside it. The graded version of \cref{eq:compositespan}, with the second system's questions written as $e_5\otimes f_j$ so that they anticommute with the first system's, has the same $13$ local generators.

This answers the title question in physical terms. Nature does contain two-level systems whose algebra carries four unbiased questions, pairs of fermionic modes; a superselection rule removes the odd sector, degrees one and three, leaving the even one with its three questions per parity block. A four-question two-level system is a fermion pair with that superselection inverted. It could not arise as an effective description either: M\"uller and Garner~\cite{MullerGarner2023} show that the only quantum-embeddable models a physical decoherence map produces are classical theory and complex quantum theory with superselection rules.

\textit{Discussion}
The algebraic, geometric and physical answers agree, and what the hyperbit lacks reduces in each of them to closure: the Jordan algebra does not close under products, the composite closes only by handing back the fifth question, and in fermionic terms both are parity superselection.

Degree two of the two-qubit grading originates the generators, the KD imaginary parts, and the fermionic covariance matrix all at once. The qubit's Hodge star folds it into degree one, which is why the qubit's KD imaginary part is the third question and its Hamiltonians are observables; nothing folds back for the hyperbit, and the substrate sets that imaginary part to either sign. Any foundational argument resting on the sign or size of KD imaginarity, or so on imaginary weak values, is therefore using the qubit's complex structure rather than the geometry of complementary questions. Whether the 24-cell's contextuality is a resource for the two-qubit states and measurements that carry it we leave open.

\textit{Acknowledgments:}
The author acknowledges support from a Royal Society Research Grant (RG/R1/251590), an EPSRC Mathematical Sciences Small Grant (UKRI3647), and from their EPSRC Quantum Technologies Career Acceleration Fellowship (UKRI1217).

\textit{Data and code availability:} A Wolfram Mathematica notebook reproducing all numerical results is available from the author on request.

\bibliographystyle{apsrev4-2}
\bibliography{ref}

\appendix

\section{Quasiprobability representations}
\label{app:wigner}

Frames give quasiprobability representations~\cite{FerrieEmerson2008,Spekkens2008}. Let $\vn_\lambda$ be $N$ unit vectors spanning $\R^n$, the index $\lambda$ running over a set closed under $\lambda\mapsto-\lambda$ with $\vn_{-\lambda}=-\vn_\lambda$, with $\sum_\lambda\vn_\lambda=0$ and $\sum_\lambda\vn_\lambda\vn_\lambda^{\mathsf T}=(N/n)\1$, as the 16-cell, tesseract and 24-cell all are. The self-dual representation is $W_\lambda(\vx)=(1+\sqrt n\,\vn_\lambda\cdot\vx)/N$ for states and $\varepsilon_\lambda(E)=a+\sqrt n\,\vb\cdot\vn_\lambda$ for effects, with $\sum_\lambda W_\lambda\varepsilon_\lambda(E)=E(\omega)$ and $\varepsilon_\lambda(E^+_\vm)=\tfrac N2W_\lambda(\vm)$, so states and sharp effects share one function and the factor $\sqrt n$ is forced. For $n=3$ and the tetrahedron it is Wootters' discrete Wigner function~\cite{Wootters1987,GHW2004}, with minimum $(1-\sqrt3)/4$. The frame of $2^{n-1}$ even-parity strings $\vs/\sqrt n$, which at $n=3$ is that tetrahedron, satisfies the same identities at every $n\ge3$; for $n=4$ those strings are frame $\mathcal B$, so the Wootters-type and the MUB-frame constructions coincide up to triality, whereas for $n=3$ they give different functions on four and six points.

\begin{theorem}
\label{thm:wigner}
On the axis frame $\vn_{\pm i}=\pm\ve_i$, $W_{\pm i}(\vx)=(1\pm\sqrt n\,x_i)/2n$ is nonnegative if and only if $|x_i|\le1/\sqrt n$ for all $i$, and the pure states with nonnegative $W$ are the $2^n$ Hadamard states $\vs/\sqrt n$ and no others. For $n=3$ these are the cube, the two SIC tetrahedra. For $n=4$ they are the tesseract $\mathcal B\cup\mathcal C$, so the Wigner function $W^{(\mathcal F)}$ of any frame $\mathcal F\in\{\mathcal A,\mathcal B,\mathcal C\}$ is nonnegative on the other two frames and their convex hull, and nowhere else. On its own frame $W^{(\mathcal F)}_\lambda(\vn_\mu)$ takes the values $\tfrac38,-\tfrac18,\tfrac18$ ($\lambda=\mu$, $\lambda=-\mu$, the other six); on a state of another frame it is $\tfrac14$ on four points and $0$ on four.
\end{theorem}

\begin{proof}
A pure state has $\sum_ix_i^2=1$, which with $|x_i|\le1/\sqrt n$ forces $|x_i|=1/\sqrt n$. The rest is evaluation.
\end{proof}

The $n=3$ case is the symmetric member of Wallman and Bartlett's family of qubit representations with four nonnegative bases~\cite{WallmanBartlett2012}, whose Bloch vectors must be the vertices of a right cuboid. The Hadamard states are also the optimal encodings of an $(n,1)$ random access code with one $n$-ball \cite{Ambainis2002,PawlowskiWinter2012}: any of $n$ bits is recovered with probability $\tfrac12(1+1/\sqrt n)$, which is $\tfrac34$ for the hyperbit, and $\sum_i\av{e_i}^2\le1$ makes this the most any encoding achieves.

\section{Why the ball}
\label{app:ball}

The ball is not the only two-level state space with $n$ unbiased questions. Describe a state by the biases $\vx=(\av{e_1},\dots,\av{e_n})$ of $n$ questions, and take as candidates the convex bodies symmetric under the hyperoctahedral group $B_n$ (signed permutations of the axes) whose only states sharp for the $i$th question are the axis points $\pm\ve_i$. The cross-polytope $\sum_i|x_i|\le1$ qualifies, with only the $2n$ axis points as pure states: one question sharp at a time and nothing else. The hypercube $[-1,1]^n$, for $n=2$ the square bit of boxworld~\cite{Barrett2007}, does not: a whole facet is sharp for a given question, and its vertices answer all $n$ questions at once.

\begin{proposition}
\label{prop:whyball}
Let $K$ be such a body in $\R^n$, $n\ge3$, with the axis points extreme and the linear maps preserving $K$ as its reversible transformations. If some reversible transformation connected to the identity moves a pure state, $K$ is the unit ball.
\end{proposition}

\begin{proof}
The maps preserving $K$ form a compact group containing $B_n$, so they preserve an inner product, the standard one because $B_n$ acts irreducibly on $\R^n$, and lie in $O(n)$. The Lie algebra of the identity component is then a $B_n$-invariant subspace of $\sog(n)\cong\Lambda^2\R^n$, which is irreducible for $n\ge3$: the sign changes act on $\ve_i\wedge\ve_j$ by the characters $s_is_j$, distinct for distinct pairs, so an invariant subspace is spanned by a subset of the $\ve_i\wedge\ve_j$, and the permutations act transitively on pairs. The component is therefore trivial or $SO(n)$. In the second case the orbit of $\ve_1$ is the unit sphere, so every point of the sphere is extreme in $K$, and no point outside the ball lies in $K$, since its orbit would put the sphere in the interior.
\end{proof}

The cross-polytope's pure states are permuted only by the finite group $B_n$, so \cref{prop:whyball} leaves only the ball.

\section{Proofs}
\label{app:proofs}

\begin{proposition}
\label{prop:16cell}
The cross-polytope subtheory of the $n$-ball, with states $\pm\ve_i$, the $n$ axis questions and their mixtures, has a preparation- and measurement-noncontextual model: ontic space $\{\pm1\}^n$, the $i$th question reading the $i$th sign, and the state $s\ve_i$ uniform on the $2^{n-1}$ ontic states with $i$th sign $s$. It reproduces $p(\pm|\ve_j;s\ve_i)=\tfrac12(1\pm s\delta_{ij})$ and the subtheory's operational equivalences.
\end{proposition}

\begin{proof}[Proof of \cref{prop:count}]
Two anticommuting invertible operators satisfy $A=-BAB^{-1}$, so $A$ is traceless and $d$ is even unless $n\le1$. A representation of $\Cl_n$ on $\C^d$ is a direct sum of irreducibles, of dimension $2^{\lfloor n/2\rfloor}$, so $2^{\lfloor n/2\rfloor}$ divides $d$ and $n\le2q+1$; tensoring an irreducible representation with $\1_m$ attains the bound. The bound is Shapiro's \cite[Ch.~1, Ex.~12]{Shapiro2000}; see \cite{Hrubes2016} for a recent account.
\end{proof}

\begin{proof}[Proof of \cref{prop:projection}]
Since $(\vn\cdot\ve)^2=|\vn|^2\1$ and each $e_i$ is traceless, $\vn\cdot\ve$ has eigenvalues $\pm1$ with multiplicity two, so $\tfrac12(\1\pm\vn\cdot\ve)$ are rank-two projectors and $\tr(\rho\,\tfrac12(\1\pm\vn\cdot\ve))=\tfrac12(1\pm\vn\cdot\pi(\rho))=E^\pm_\vn(\pi(\rho))$. With $\vn=\vx/|\vx|$ this gives $|\vx|=\tr(\rho\,\vn\cdot\ve)\le1$, and for $\psi$ in the $+1$ eigenspace of $\vn\cdot\ve$ and $\vm\perp\vn$, anticommutation gives $\bra\psi\vm\cdot\ve\ket\psi=0$, so $\pi(\psi)=\vn$ and mixtures fill the ball. Finally the $e_ie_j$ span $\sog(4)$, and $\exp(-\theta e_ie_j)e_k\exp(\theta e_ie_j)$ rotates $e_i,e_j$ by $2\theta$ and fixes the rest: this is $\Spin(4)\cong SU(2)\times SU(2)\to SO(4)$, sitting inside $\Spin(5)\cong\mathrm{Sp}(2)\subset U(4)$ with $\mathrm{Sp}(2)$ the $2\times2$ quaternionic unitaries. Conjugation by $e_i$ sends $e_j\mapsto-e_j$ for $j\neq i$.
\end{proof}

\begin{proof}[Proof of \cref{thm:kdrange}]
$e_i,e_j,G_{ij}$ pairwise anticommute, so $x_i^2+x_j^2+g_{ij}^2\le1$ as in the proof of \cref{prop:projection}. For attainment use the factorisation $\C^4=\C^2\otimes\C^2$ with $e_i=\sigma_x\otimes\1$, $e_j=\sigma_y\otimes\1$, $G_{ij}=\sigma_z\otimes\1$, $e_k=\sigma_z\otimes\sigma'_x$, $e_l=\sigma_z\otimes\sigma'_y$ for the remaining indices, primed operators acting on the second factor; any four anticommuting involutions on $\C^4$ are unitarily equivalent to these. A product state with Bloch vectors $\mathbf a,\mathbf b$ has $x_i=a_x$, $x_j=a_y$, $g_{ij}=a_z$, $x_k=a_zb_x$, $x_l=a_zb_y$. Given $\vx$, take $a_z=\pm\sqrt{1-x_i^2-x_j^2}$ and $\mathbf b=(x_k,x_l,0)/a_z$, a valid Bloch vector since $x_k^2+x_l^2\le1-x_i^2-x_j^2$; the two signs give the same $\vx$ with opposite $g_{ij}$, and mixing them with weights $p,1-p$ gives $g_{ij}=(2p-1)\sqrt{1-x_i^2-x_j^2}$.
\end{proof}

\begin{proof}[Proof of \cref{prop:compositeprops}]
Local tomography holds because the observables are spanned by products, and $25=5\times5$; no-signalling because $e_i\otimes\1$ and $\1\otimes f_j$ commute. The extreme points of the trace-norm unit ball are the rank-one matrices $\mathbf u\mathbf v^{\mathsf T}$ with unit vectors, so the convex hull of $\{\vx\vy^{\mathsf T}:|\vx|,|\vy|\le1\}$ is that ball. The $e_i\otimes f_i$ commute, square to $\1$, and have product $e_5\otimes f_5$; their joint $+1$ eigenspace is one-dimensional. For any orthogonal $O$, $\sum_i(O\ve)_i\otimes f_i$ is again a sum of four commuting involutions, so $\tr(O^{\mathsf T}C)\le4$. Marginals vanish because $e_i\otimes\1$ anticommutes with $e_j\otimes f_j$ for $j\neq i$. The CHSH value follows from $C=\1_4$; the bound holds because $\vn\cdot\ve\otimes\1$ and $\1\otimes\vm\cdot\vf$ are commuting Hermitian involutions on a four-qubit state \cite{Tsirelson1980}.
\end{proof}

\begin{proof}[Proof of the bound in \cref{eq:sbound}]
Let $\mu_k$, $\mu_K$ be the ontic distributions of an antipodal pair and $\nu$ that of the maximally mixed state. Preparation noncontextuality turns the hexagon's equivalences into $\mu_k+\mu_K=2\nu$ for each pair and $\mu_a+\mu_b+\mu_c=3\nu$. Where $\nu$ vanishes so do all six $\mu$, so writing $\mu_k=2\nu u_k$ defines $u_k\in[0,1]$ up to a $\nu$-null set, with $\mu_K=2\nu(1-u_k)$ and $\sum_ku_k=\tfrac32$. With response functions $\xi_k$,
\begin{equation*}
\begin{aligned}
s&=\tfrac13\textstyle\int\nu\sum_k\big[\xi_ku_k+(1-\xi_k)(1-u_k)\big]\\
&\le\tfrac13\textstyle\int\nu\sum_k\max(u_k,1-u_k),
\end{aligned}
\end{equation*}
and $\sum_k\max(u_k,1-u_k)=\tfrac32+\sum_k|u_k-\tfrac12|$. The deviations $u_k-\tfrac12$ sum to zero and lie in $[-\tfrac12,\tfrac12]$, so their absolute values sum to at most $1$ and $s\le\tfrac56$.
\end{proof}

\section{Symmetries of the composite}
\label{app:local}

\begin{lemma}
\label{lem:pauli}
Let $\mathcal V$ be the real span of a set $S$ of Pauli operators on any number of qubits, closed under sign. A Hermitian $H=\sum_kh_kP_k$ in the Pauli basis satisfies $i[H,v]\in\mathcal V$ for all $v\in\mathcal V$ if and only if $h_k=0$ for every $P_k$ that anticommutes with some $v\in S$ with $iP_kv\notin\mathcal V$. The Lie algebra of unitaries preserving $\mathcal V$ is spanned by the $iP$ with $P$ Pauli such that, for every $v\in S$, $[P,v]=0$ or $iPv\in\mathcal V$.
\end{lemma}

\begin{proof}
For fixed $v\in S$, $i[P_k,v]$ is $0$ or $2iP_kv$, Hermitian in the second case, and $P_k\mapsto P_kv$ is a bijection of the Pauli basis up to phases. So the Pauli components of $i[H,v]$ are the numbers $2h_k$ at the distinct positions $iP_kv$, and $i[H,v]\in\mathcal V$ holds if and only if each position outside $\mathcal V$ carries $h_k=0$. With multiplicity spaces, $H=\sum_kP_k\otimes M_k$, the same argument forces $M_k=0$ for a failing $P_k$ and $M_k\propto\1$ for a passing $P_k$ that anticommutes with some $v$; only the identity commutes with all of $\mathcal O_{4\otimes4}$.
\end{proof}

For $\mathcal V=\mathcal O_{4\otimes4}$ the Pauli operators passing the test are the identity and the twelve giving the local bivectors $e_ie_j\otimes\1$ and $\1\otimes f_kf_l$, and no others, which is \cref{thm:localonly}.

The same test shows that a realisation with multiplicity, $e_i\otimes\1_m$ on $\C^4\otimes\C^m$, adds only generators on the multiplicity spaces, which do not change the effective statistics, so the conclusion is realisation-independent; and that enlarging the local observable space does not help, since adding the fifth question on both sides (symmetry algebra of dimension $21$), adding the generators as well ($13$), adding everything but the fifth question ($15$), or replacing one side by a full pair of qubits ($22$) all leave that algebra local. For a product generator $A\otimes f_j$ the obstruction can be seen by hand: acting on $e_i\otimes f_k$ ($k\neq j$) it produces $\{A,e_i\}\otimes f_jf_k$, whose second factor has degree two, so $\{A,e_i\}$ must vanish for all $i$, forcing $A\propto e_5$; and then $[e_5,e_i]$ has degree three, outside $\mathcal O_4$. A qubit encoded as $\{\sigma_i\otimes X\}$ hits the same obstruction, once the generators commuting with the whole span are set aside, and one encoded as $\{\sigma_i\otimes\1\}$ does not.

\end{document}